\documentclass[11pt]{amsart}

\usepackage{amssymb,amscd,amsmath,amsthm}
\usepackage{graphicx}
\newtheorem{theorem}{Theorem}[section]

\newtheorem{lemma}[theorem]{Lemma}
\newtheorem{definition}[theorem]{Definition}

\newtheorem{remark}[theorem]{Remark}

\def\ca{\mathcal{ A}}

\def\ch{{\mathcal H}}

\def\ck{{\mathcal K}}

\def\cm{{\mathcal M}}

\def\bh{{\mathbb H}}

\def\bn{{\mathbb N}}
\def\bp{{\mathbb P}}

\def\bz{{\mathbb Z}}

\def\frak{\mathfrak}

\def\ga{{\frak A}}

\def\f{\varphi}
\def\r{\rho}
\def\O{\Omega}

\def\tr{{\rm Tr}}
\def\L{\Lambda}
\def\G{\Gamma}

\def\ffi{\varphi}

\def\Tr{\mathrm{Tr}}
\def\<{\langle}
\def\>{\rangle}

\def\1{\mathbf{1}}

\def\ve{\varepsilon}

\def\LL{\Lambda}
\def\o{\otimes}

\def\bh{\mathbf{h}}
\def\bs{\mathbf{s}}

\def\id{{\bf 1}\!\!{\rm I}}

\begin{document}

\begin{center}
{\Large {\bf Open Quantum Random Walks and Quantum Markov
chains on Trees I: Phase transitions }}\\[1cm]

{\sc Farrukh Mukhamedov} \\[2mm]

 Department of Mathematical Sciences, College of Science, \\
 United Arab Emirates University 15551, Al-Ain,\\ United
Arab Emirates\, and \\
Institute of Mathematics named after V.I.Romanovski, 4,\\
University str., 100125, Tashkent, Uzbekistan\\
e-mail: {\tt far75m@gmail.com; farrukh.m@uaeu.ac.ae}\\[1cm]

{\sc Abdessatar Souissi} \\[2mm]

Department of Accounting, College of Business Management\\
Qassim University, Ar Rass, Saudi Arabia \, and \\
Preparatory institute for scientific and technical studies,\\
 Carthage University, Amilcar 1054, Tunisia\\
  e-mail: {{\tt a.souaissi@qu.edu.sa}; {\tt
 abdessattar.souissi@ipest.rnu.tn}}\\[1cm]

 {\sc Tarek Hamdi} \\[2mm]

Department of Management Information Systems, College of Business Management\\
Qassim University, Ar Rass, Saudi Arabia \, and \\
Laboratoire d'Analyse Math\'ematiques et applications LR11ES11 \\
Universit\'e de Tunis El-Manar, Tunisia\\
  e-mail: {{\tt t.hamdi@qu.edu.sa}}\\[1cm]
\end{center}

\small
\begin{center}
{\bf Abstract}\\
\end{center}
In the present paper, we construct QMC (Quantum Markov Chains)
associated with Open Quantum Random Walks such that the transition
operator of the chain is defined by OQRW and the restriction of QMC
to the commutative subalgebra coincides with the distribution
$\bp_\r$ of OQRW. However, we are going to look at the probability distribution as a Markov field over the Cayley tree.
Such kind of consideration allows us to investigated phase transition phenomena associated for OQRW within QMC scheme.
Furthermore, we first propose a new construction of QMC on trees, which is an extension of QMC considered in Ref. \cite{AOM}. Using such a construction, we are able to construct QMCs on tress associated with OQRW. Our investigation leads to the detection of the phase transition phenomena within the proposed scheme. This kind of phenomena appears first time in this direction. Moreover, mean entropies of QMCs are calculated.


\normalsize

\section{Introduction}

Discovering the aspects of quantum mechanics, such as
superposition and interference, has lead to the idea of
quantum walks; a generalization of classical random walks \cite{D,Ke,KOW11,M}.
Recently, in \cite{CL08} a quantum phase transition has been explored by means of quantum walks in an optical lattice.
On the other hand, in \cite{Ka1} it has been shown that discrete-time quantum walks (QW)
can realise topological phases in 1D and 2D for all the symmetry classes
of free-fermion systems. In particular, they provide the QW protocols that simulate
representatives of all topological phases, featured by the presence of robust symmetry-
protected edge states \cite{Ka2}. In general, QW realisations are particularly useful,
because, in addition to the simplicity of their mathematical description, the parameters
that define them can be easily controlled in the lab.

Over the past decade, motivated
largely by the prospect of superefficient algorithms, the theory of quantum Markov
chains (QMC), especially in the guise of quantum walks, has generated a huge number of works,
including many discoveries of fundamental importance \cite{AW87,DM19,Ke,Kum,portugal}.
In \cite{Fing} a
novel approach  has been proposed to investigate quantum cryptography problems by means of QMC \cite{Gud}, where quantum effects are entirely encoded into super-operators labelling transitions, and
the nodes of its transition graph carry only classical information and thus they are discrete. QMC have been applied  \cite{AW87,DK19,DM19,DM191} to the investigations of so-called "open quantum random walks" (OQRW) \cite{attal,carbone,carbone2,konno,cfrr}.
We notice that OQRW are related to the study of asymptotic behavior of trace-preserving completely
positive maps, which belong to fundamental
topics of quantum information theory ( see for instance
\cite{burgarth2,petulante,novotny}).

For the sake of clarity, let us recall some necessary information about OQRW.
 Let $\mathcal{K}$ denote a separable Hilbert space and let
$\{|i\rangle\}_{i\in \LL}$ be its orthonormal basis indexed by the
vertices of some graph $\LL$ (here the set $\LL$ of vertices might
be finite or countable). Let $\mathcal{H}$ be another Hilbert space,
which will describe the degrees of freedom given at each point of
$\LL$. Then we will consider the space
$\mathcal{H}\otimes\mathcal{K}$. For each pair $i,j$ one associates
a bounded linear operator $B_{j}^i$ on $\mathcal{H}$. This operator
describes the effect of passing from $|j\rangle$ to $|i\rangle$. We
will assume that for each $j$, one has
\begin{equation}\label{Bij}
\sum_i B_{j}^{i*}B_{j}^i=\id,
\end{equation} where, if infinite,
such series is strongly convergent. This constraint means: the sum
of all the effects leaving site $j$ is $\id$. The operators $B^i_j$
act on $\mathcal{H}$ only, we dilate them as operators on
$\mathcal{H}\otimes\mathcal{K}$ by putting
$$
M^i_j=B^i_j\otimes \vert i\rangle\langle j\vert\,.
$$
The operator $M^i_j$ encodes exactly the idea that while passing
from $\vert j\rangle$ to $\vert i\rangle$ on the lattice, the effect
is the operator $B^i_j$ on $\mathcal{H}$.

According to \cite{attal} one has
\begin{equation}\label{sumMij=1}
\sum_{i,j} {M^i_j}^* M^i_j=\id.
\end{equation}

Therefore, the operators $(M^i_j)_{i,j}$ define a completely
positive mapping
\begin{equation}\label{MM}
\cm(\r)=\sum_i\sum_j M^i_j\,\r\, {M^i_j}^*
\end{equation}
on $\ch\otimes\ck$.

In what follows, we consider density matrices on
$\mathcal{H}\otimes\mathcal{K}$ which take the form
\begin{equation}\label{rr}
\rho=\sum_i\rho_i\otimes |i\rangle\langle i|,
\end{equation}
 assuming that $\sum_i\tr(\rho_i)=1$.

For a given initial state of such form, the \textit{Open Quantum
Random Walk (OQRW)} is defined by the mapping $\cm$, which has the
following form
\begin{equation}\label{MM1}
\cm(\rho)=\sum_i\Big(\sum_j B_{j}^i\rho_j B_{j}^{i*}\Big)\otimes
|i\rangle\langle i|.
\end{equation}

By means of the map $\cm$ one defines a family of classical random
process on $\O=\L^{\bz_+}$. Namely, for any density operator $\r$ on
$\mathcal{H}\otimes\mathcal{K}$ (see \eqref{rr}) the probability
distribution is defined by
\begin{equation}\label{Prb}
\bp_\r(i_0,i_1,\dots,i_n)=\tr(B^{i_n}_{i_{n-1}}\cdots
B^{i_2}_{i_{1}}B^{i_1}_{i_{0}}\rho_{i_0}B^{i_1*}_{i_{0}}B^{i_2*}_{i_{1}}\cdots
B^{i_n*}_{i_{n-1}}).
\end{equation}
We point out that this distribution is not a Markov measure
\cite{BBP}.

On the other hand, it is well-known \cite{AW87,norris} that to each
classical random walk one can associate a certain Markov chain and
some properties of the walk can be explored by the constructed
chain. Therefore, it is natural to construct a Quantum Markov chain (QMC)
associated with OQRW and investigate its properties.

Recently, in \cite{DK19,DM19}, we have found a quantum Markov
chain \footnote{We note that a Quantum Markov Chain is a
quantum generalization of a Classical Markov Chain where the state
space is a Hilbert space, and the transition probability matrix of a
Markov chain is replaced by a transition amplitude matrix, which
describes the mathematical formalism of the discrete time evolution
of open quantum systems, see
\cite{[Ac74f]}-\cite{[AcFr80]},\cite{fannes2,FM,nacht} for more details.}
(or finitely correlated state \cite{fannes2}) $\f$ on the
algebra $\ca=\otimes_{i\in\bz_+}\ca_i$, where $\ca_i$ is isomorphic to
$B(\ch)\otimes B(\ck)$, $i\in\bz_+$, such that the transition operator
$P$ equals to the mapping
$\cm^*$\footnote{The dual of $\cm$ is defined by the equality
$\tr(\cm(\r)x)=\tr(\r\cm^*(x))$ for all density operators $\rho$ and
observables $x$.} and the restriction of $\f$ to the commutative
subalgebra of $\ca$ coincides with the distribution $\bp_\r$, i.e.
\begin{equation}\label{I1}
\f\big((\id\o|i_0><i_0|)\o\cdots\o(\id\o|i_n><i_n|)\big)=\bp_\r(i_0,i_1,\dots,i_n).
\end{equation}
Hence, this result allows us  to
interpret the distribution $\bp_\r$ as a QMC, and to study further
properties of $\bp_\r$.

In the present paper, we initiate to look at the probability distribution \eqref{Prb} as a Markov field over the Cayley tree $\G^k$ \cite{D}.
Roughly speaking, $(i_0,i_1,\dots,i_n)$ is considered as a configuration on $\O=\L^{\G^k}$. Such kind of consideration allows us to investigated phase transition phenomena associated for OQRW within QMC scheme \cite{MBS161,MBSG20}. We stress that, in physics, a spacial classes of QMC, called "Matrix Product States" (MPS) and more generally "Tensor Network States" \cite{CV,Or} were used to investigate
quantum phase transitions for several lattice models.
This method uses the density matrix renormalization group (DMRG) algorithm which opened a new way of performing
the renormalization  procedure in 1D systems and gave extraordinary precise results. This is done by keeping the states of subsystems which
are relevant to describe the whole wave-function, and not those that minimize the energy on
that subsystemc \cite{[RoOs96]}.

In \cite{AccMuSa1,AccMuSa2,AccMuSa3,AOM,MG17,MG19} a QMC approach has been used to investigate models defined over the Cayley trees.
In this path the QMC scheme is based on the
 $C^*$-algebraic framework (see also \cite{[AcSouElG20],MG20}).
Furthermore, in \cite{MBS161,MBS162,MBSG20,MR1,MS19} we have established that Gibbs measures of the Ising model with competing (Ising) interactions (with commuting interactions) on a Cayley trees, can be considered as  QMC (see also \cite{MS20}).

In this paper, we first propose a new construction of QMC on trees, which is an extension of QMC considered in \cite{AOM,DM19,Park}. Using such a construction, we are able to construct QMC on tress associated with OQRW. Furthermore, our investigation leads to the detection of the phase transition phenomena within the proposed scheme. This kind of phenomena appears first time in this direction. Moreover, mean entropies of QMCs are calculated (cp. \cite{OW19,Watan11}). We point out that, recently, in \cite{Mach21} 1st and 2nd moments of the open quantum walk have been studied and found its standard deviation. A phase transition is observed by evaluating the standard deviation, i.e. whether the quantum walk has diffusive or ballistic behavior.

\section{Preliminaries}\label{sec_prel}

Let $\Gamma^k_{+} = (V,E)$ be the semi-infinite Cayley tree of order $k$  with root  $o$. The Cayley tree of order $k$ is characterized by being a tree for which every vertex has exactly
 $k+1$ nearest-neighbors ( see \cite{ MBS161}). Recall that, two vertices $x$ and $y$ are  {\it nearest neighbors} (denoted $x\sim y$ ) if they are joined through an edge (i.e. $<x,y>\in E$). A   list $ x\sim x_1\sim \dots \sim x_{d-1}\sim y$ of vertices is called a {\it
path} from $x$ to $y$. The distance on the tree between two vertices $x$ and $y$ (denoted $d(x,y)$) is the length of the shortest edge-path
joining them.

Define
\[W_n := \{x\in V \quad \mid\quad d(x,o) = n \}\]
\[ \Lambda_{n}: = \bigcup_{j\le n}W_j;\quad  \Lambda_{[m,n]} = \bigcup_{j=m}^{n}W_j.\]

Recall a coordinate structure in $\G^k_+$:  every vertex $x$
(except for $x^0$) of $\G^k_+$ has coordinates $(i_1,\dots,i_n)$,
here $i_m\in\{1,\dots,k\}$, $1\leq m\leq n$ and for the vertex
$x^0$ we put $(0)$.  Namely, the symbol $(0)$ constitutes level 0,
and the sites $(i_1,\dots,i_n)$ form level $n$ (i.e. $d(x^0,x)=n$)
of the lattice.  Using this structure, vertices
$x^{(1)}_{W_n},x^{(2)}_{W_n},\cdots,x^{(|W_n|)}_{W_n}$ of $W_n$
can be represented as follows:
\begin{eqnarray}\label{xw}
&&x^{(1)}_{W_n}=(1,1,\cdots,1,1), \quad x^{(2)}_{W_n}=(1,1,\cdots,1,2), \ \ \cdots \quad x^{(k)}_{W_n}=(1,1,\cdots,1,k,),\\
&&x^{(k+1)}_{W_n}=(1,1,\cdots,2,1), \quad
x^{(2)}_{W_n}=(1,1,\cdots,2,2), \ \ \cdots \quad
x^{(2k)}_{W_n}=(1,1,\cdots,2,k),\nonumber
\end{eqnarray}
\[\vdots\]
\begin{eqnarray*}
&&x^{(|W_n|-k+1)}_{W_n}=(k,k,,\cdots,k,1), \
x^{(|W_n|-k+2)}_{W_n}=(k,k,\cdots,k,2),\ \ \cdots
x^{|W_n|}_{W_n}=(k,k,\cdots,k,k).
\end{eqnarray*}
In the above notations, we write
$$
 W_n = \{ (i_1, i_2, \cdots, i_n); \quad i_j = 1,2, \cdots, k \}
$$
So one can see that $|W_n|=k^n$.
The set of \textit{direct successors}  for a given vertex $x\in V$  is defined  by
\begin{equation}\label{S(x)def}
S(x) :  = \left\{y\in V \, \,  : \, \,  x\sim y \, \, \hbox{and} \, \, d(y,o) > d(x,o) \right\}.
\end{equation}
The vertex $x$ has exactly $k$ direct successors denoted $(x,i), i=1,2,\cdots, k$
$$
S(x) = \{(x,1),\,  (x,2), \, \cdots, \, (x,k)\}.
$$

To each vertex $x$, we associate a C$^*$--algebra of observable $\mathcal{A}_x$ with identity $\id_x$. For a given bounded region $V'\subset V$, we consider the algebra $\mathcal{A}_{V'} = \bigotimes_{x\in V'}\mathcal{A}_x$. We have the  the following natural embedding
$$
\mathcal{A}_{ \Lambda_{n }}\equiv  \mathcal{A}_{ \Lambda_{n }}\otimes\id_{ W_{n+1}}\subset \mathcal{A}_{ \Lambda_{n+1 }}.
$$
The algebra $\mathcal{A}_{ \Lambda_{n}}$ is then a subalgebra of $\mathcal{A}_{ \Lambda_{n+1}}$. It follows the local algebra
\begin{equation}\label{AVloc}
  \mathcal{A}_{V;\, loc} := \bigcup_{n\in\mathbb{N}}\mathcal{A}_{\Lambda_{n}}
\end{equation}
and the quasi-local algebra
$$
\mathcal{A}_V := \overline{\mathcal{A}_{V;\, loc}}^{C^*}
$$
The set of states on a C$^*$--algebra $\mathcal{A}$ will be denoted $\mathcal{S}(\mathcal{A})$.

Consider a triplet $\mathcal C\subseteq\mathcal B\subseteq \mathcal A$ of C$^\ast$--algebras. A \textit{quasi-conditional expectation} \cite{ACe} is a completely positive identity preserving linear map $E :\mathcal A \to \mathcal B$ such that
$
E(ca) = cE(a)$, for all $a\in\mathcal A$, $c\in \mathcal C.
$
\begin{definition}\cite{ACe}
Let $\mathcal{B}\subseteq \mathcal{A}$ be two unitary C$^*$--algebra $\id$. A Markov transition expectation from $\mathcal{A}$ into $\mathcal{B}$ is a completely positive identity preserving map.
\end{definition}

 \begin{definition}\label{QMCdef}\cite{[AcSouElG20]} A {\it (backward) quantum Markov chain}  on $\mathcal{A}_V$
 is a triplet $(\phi_o, (E_{\Lambda_{n}})_{n\ge 0},  (h_{n})_n)$
of initial state $\phi_o\in \mathcal{S}(\mathcal{A}_o)$, a sequence of quasi-conditional expectations $(E_{ \Lambda_{n }})_n$ w.r.t.
the triple $\mathcal{A}_{{\Lambda}_{n-1 }}\subseteq \mathcal{A}_{ \Lambda_{n }}\subseteq\mathcal{A}_{ \Lambda_{n+1}}$ and
a sequence $h_{n}\in\mathcal{A}_{W_n, +}$ of boundary conditions such that for each $a\in \mathcal{A}_V$  the limit
\begin{equation}\label{lim_Mc}
\varphi(a): = \lim_{n\to\infty} \phi_0\circ E_{ \Lambda_{0}}\circ
E_{ \Lambda_{1}} \circ \cdots \circ E_{ \Lambda_{n}}(h_{n+1}^{1/2}ah_{n+1}^{1/2})
\end{equation}
exists in the weak-*-topology and defines a state. In this case the state $\varphi$ defined by (\ref{lim_Mc})
is also called quantum Markov chain (QMC).
\end{definition}

\begin{remark}
The above definition introduces  quantum Markov chains on trees
 as a triplet generalizing the definitions considered in \cite{[AcFiMu07],AOM,[AMSo],MBS161}
 by adding the boundary conditions.
 On the other hand, it extends to trees the recent unifying definition for
quantum Markov chains on the one-dimensional case \cite{[AcSouElG20]}.
\end{remark}

 \section{QMC associated with OQRW on trees}\label{QMC_tree}

 Let $\mathcal{H}$ and $\mathcal{K}$ be two separable Hilbert spaces. Let   $\{ |i\rangle\}_{i\in \Lambda}$ be an  ortho-normal basis of $\mathcal{K}$ indexed by a graph $\Lambda$. To each $x\in V$  we associate the algebra  $\mathcal{A}_x \equiv \mathcal{A}:= \mathcal{B}(\mathcal{H}\otimes\mathcal{K})$.

Let $\cm$ be a OQRW given by \eqref{MM1}. In this section we will
use notations from the previous sections.

As before, for each $(i,j)\in\Lambda^2,$ one associates an operator $B_{j}^{i}\in \mathcal{B}(\mathcal{H})$ to describe the transition from the state $|j\rangle$ to the state $|i\rangle$ such that
 \begin{equation}\label{sumBB=1}
 \sum_{i\in\Lambda}B_j^{i*}B_{j}^{i} = \id_{\mathcal{B}(\mathcal{H})}.
 \end{equation}
 Consider the density operator $\rho\in \mathcal{B}(\mathcal{H}\otimes\mathcal{K})$, of the form
 $$
 \rho = \sum_{i\in\Lambda}\rho_i\otimes |i\rangle\langle  i|; \quad \rho_i\in\mathcal{B}(\mathcal{H})^{+}.
 $$
 In what follows, for the sake of simplicity of calculations,  we assume that $\r_i\neq 0$ for all $i\in\Lambda$ (see \cite[Remark 4.5]{DM19} for other kind of initial states).

 Let us consider
 \begin{equation}\label{Mij}
 M_j^i = B_j^i\otimes|i\rangle\langle j| \in\mathcal{B}(\mathcal{H}\otimes \mathcal{K}).
 \end{equation}
 Put
 \begin{equation}\label{Aij}
 A_{j}^{i} := \frac{1}{{\Tr(\rho_j)}^{1/2}}\rho_j^{1/2}\otimes |i\rangle\langle j|, \quad i,j\in \Lambda.
 \end{equation}
For each $u\in V$, we set
\begin{equation}\label{Kji}
 {K_{j}^{i}}^{(\{u\}\cup S(u))} := {M_{j}^{i*}}^{(u)}\otimes\bigotimes_{v\in S(u)}{A_{j}^{i}}^{(v)} \in\mathcal{A}_{\{u\}\cup S(u)}.
\end{equation}

Put
\begin{equation}\label{Eu_def}
\mathcal{E}_u(a)
:=\sum_{(i,j), (i',j')\in\Lambda^2} \tr_{u]}\left( {K_{j}^{i}}^{(\{u\}\cup S(u))}a{ K_{j'}^{i'}}^{(\{u\}\cup S(u)),*}\right); \quad a\in\mathcal{A}_{\{u\}\cup S(u)}.
\end{equation}
For the sake of shortness, if no confusion is caused, the operator ${K_{j}^{i}}^{(\{u\}\cup S(u))}$ will be denoted simply by ${K_{j}^{i}}$.
\begin{lemma}\label{lemmaEu}
 For each $u\in V,$ the map $\mathcal{E}_{u}$ defines a Markov transition expectation from $\mathcal{A}_{\{u\}\cup S(u)}$ into $\mathcal{A}_u$. Moreover, we have
\begin{equation}\label{Eu}
    \mathcal{E}_{u}(a_{0}^{(u)}\otimes a_1^{(u,1)}\otimes\cdots\otimes a_k^{(u,k)}) =   \sum_{(i,j ,j')\in \Lambda^3} \left(\prod_{\ell=1}^{k} \varphi_{j,j'}(a_{\ell}^{(u,\ell)})\right) M_j^{i*} a_{0}^{(u)}M_{j'}^{i}
\end{equation}
where
\begin{equation}\label{phijj'}
\varphi_{jj'}(b):= \frac{1}{\Tr(\rho_j)^{1/2}\Tr(\rho_{j'})^{1/2}}\Tr\left(\rho_j^{1/2} \rho_{j'}^{1/2} \otimes |j'\rangle\langle j|\, b\right); \quad \forall b\in\mathcal{B}(\mathcal{H})\otimes \mathcal{B}(\mathcal{K})
\end{equation}
for every $a_0, a_1, \cdots, a_k\in \mathcal{B}(\mathcal{H})\otimes \mathcal{B}(\mathcal{K}).$
\end{lemma}
\begin{proof} 
Let $a =a_{0}^{(u)}\otimes a_1^{(u,1)}\otimes\cdots\otimes a_k^{(u,k)}$, according to (\ref{Eu_def}) one has
\begin{eqnarray*}
\mathcal{E}_{u}( a )
&=& \sum_{(i,j), (i',j')\in\Lambda^2} \tr_{u]}\left( {K_{j}^{i}} a_{0}^{(u,0)}\otimes a_{1}^{ (u,1)}\cdots\otimes a_{k}^{(u,k)}{ K_{j'}^{i'}}^*\right)\\
&=&  \tr_{u]}\left( \left(\sum_{(i,j) \in\Lambda^2}{K_{j}^{i}}\right) a \left(\sum_{  (i,j)\in\Lambda^2}{ K_{j}^{i}}\right)^*\right)\\
\end{eqnarray*}
Then $\mathcal{E}_u$ has a Krauss form and it is completely positive.
Taking into account (\ref{Kji}) and  (\ref{Mij}) one gets
\begin{eqnarray*}
 \mathcal{E}_{u}(a) &= &
 \sum_{(i,j),(i',j')\in \Lambda^2}\tr_{u]}\left( {K_{j}^{i}}^{(\{u\}\cup S(u))}a{ K_{j'}^{i'}}^{(\{u\}\cup S(u))*}\right)\\
&=&\sum_{(i,j),(i',j')\in \Lambda^2} M_j^{i*}a_0^{(u, 0)}M_{j'}^{i'} \prod_{\ell=1}^{k} \Tr( A_{j}^{i}a_{\ell}^{(u,\ell)}A_{j'}^{i'*}).\\
\end{eqnarray*}
From (\ref{Aij}) for each $\ell\in\{1,\dots, k\}$ one has
\begin{eqnarray*}
\Tr( A_{j}^{i}a_{\ell}^{(u,\ell)}A_{j'}^{i'*}) &=& \Tr( A_{j'}^{i'*}A_{j}^{i}a_{l}^{(u,\ell)})\\
&=& \frac{1}{\Tr(\rho_j)^{1/2}\Tr(\rho_{j'})^{1/2}}\Tr(\rho_{j'}^{1/2}\otimes |j'\rangle\langle i'|\rho_{j}^{1/2}\otimes |i\rangle\langle j|a_{\ell}^{(u,\ell)})\\
&=& \frac{1}{\Tr(\rho_j)^{1/2}\Tr(\rho_{j'})^{1/2}}\Tr(\rho_{j'}^{1/2}\rho_{j}^{1/2}\otimes |j'\rangle\langle   j|a_{\ell}^{(u,\ell)})\delta_{i,i'}\\
\end{eqnarray*}
where $\delta_{i,i'}$ denotes the Kronecker symbol. This leads to (\ref{Eu}) and finishes the proof. \end{proof}
\begin{lemma}\label{lem-ELbn}
For each $n\in\mathbb{N}$, the map
\begin{equation}\label{E-Wn}
\mathcal{E}_{W_n} = \bigotimes_{u\in W_n}\mathcal{E}_{u}
\end{equation}
defines a Markov transition expectation from $\mathcal{A}_{\Lambda_{[n,n+1]}}$ into $\mathcal{A}_{W_n}$. Moreover, the map
\begin{equation}\label{E_Lbn}
    E_{\Lambda_n} = id_{\mathcal{A}_{\Lambda_{n-1}}}\otimes \mathcal{E}_{W_n}
\end{equation}
is a quasi-conditional expectation w.r.t. the triplet $\mathcal{A}_{\Lambda_{n-1}}\subset\mathcal{A}_{\Lambda_n}\subset\mathcal{A}_{\Lambda_{n+1}}$.
\end{lemma}
\begin{proof} Thanks to the Cayley tree structure $W_{n+1} = \bigsqcup_{u\in W_n}S(u)$, where $\bigsqcup$ means the disjointedness of the union.  One gets the result using Lemma \ref{lemmaEu}.
\end{proof}
\begin{remark}
 In the notations of Definition \ref{QMCdef} the triplet $(\phi_o, (E_{\Lambda_{n}})_{n\ge 0},  (h_{n})_n)$ defining a quantum Markov chain $\varphi$ on $\mathcal{A}_V$ through (\ref{lim_Mc}) reduces to a finer triplet $(\phi_o, (\mathcal{E}_{u})_{u\in V}, (h_{u})_{u\in V})$ where $\phi_o\in\mathcal{S}(\mathcal{A}_o)$, the family of localized Markov transition expectations  $(\mathcal{E}_{u})_{u\in V}$ relates to the sequence of quasi-conditional expectations  $(E_{\Lambda_n})_n$ through (\ref{E-Wn}) and (\ref{E_Lbn}) and $h_n = \bigotimes_{u\in W_n}h_u$.
\end{remark}

\begin{theorem}\label{thm_E_expression} Let $M_{j}^{i}$ and $A_{j}^{i}$ be given by (\ref{Aij}) and (\ref{Mij}). In the notations of Lemma \ref{lemmaEu}, if an initial density matrix $\omega_o\in \mathcal{A}_{o; +}$ and a boundary condition $(h_u)_{u\in V}$ satisfy
\begin{equation}\label{TrwohoOQRW}
    \Tr(\omega_oh_o) = 1
\end{equation}
and
\begin{equation}\label{h=E(h)OQRW}
\sum_{i,j,j'\in\Lambda} M_j^{i*} M_{j'}^{i} \prod_{\ell=1}^{k} \varphi_{j,j'}(h_{(u,\ell)}) = h_u
\end{equation}
Then the triplet $(\omega_o, (\mathcal{E}_{u})_{u\in V}, (h_u)_{u\in V})$ defines a quantum Markov chain $\varphi$ on the algebra $\mathcal{A}_V$. Moreover, for each $a= \bigotimes_{u\in \Lambda_n}a_u\in\mathcal{A}_{\Lambda_n}$ one has
\begin{equation}\label{phi_id}
\varphi(a)= \sum_{ j , j'\in \Lambda} \Tr\left(\mathcal{M}_{jj'}(\omega_o) a_o \right)\prod_{u\in \Lambda_{[1,n]}}\psi_{j,j'}(a_{u})
\prod_{v\in \Lambda_{n+1}} \varphi_{j,j'}(h^{(v)})
\end{equation}
 where $\mathcal{E}_u$ is given by (\ref{Eu}), the functional $\varphi_{jj}$ is given by (\ref{phijj'}), and
\begin{equation}\label{Mjj'}
   \mathcal{M}_{jj'}( \cdot ) = \sum_{i\in\Lambda} M_{j'}^{i} \, \cdot\,M_j^{i*},
\end{equation}
\begin{equation}\label{psi_jjprime}
\psi_{j,j'}(b)= \frac{1}{\Tr(\rho_j)^{1/2}\Tr(\rho_{j'})^{1/2}}\sum_{i\in\Lambda}\Tr\left(B_{j'}^{i}\rho_{j'}^{1/2} \rho_{j}^{1/2}{B_{j}^{i}}^{*} \otimes |i\rangle\langle i|\, b\right).
\end{equation}
\end{theorem}

\begin{proof} Let us first prove the existence of the Markov chain $\varphi$ by examining the limit (\ref{lim_Mc}) for $E_{\Lambda_n}$ as in Lemma \ref{lem-ELbn} and $h_n =  \bigotimes_{u\in W_n}h_u$.
From Lemma \ref{lemmaEu}  the equality (\ref{h=E(h)OQRW}) is equivalent to
$$
h_u = \mathcal{E}_u(\id_{u}\otimes h_{(u,1)}\otimes \cdots h_{(u,k)})
$$
It follows that for each integer $m\in\mathbb{N}$ one has
\begin{equation}\label{Ewm(h)}
\mathcal{E}_{W_m}(\id_{W_n}\otimes h_{m+1}) = \bigotimes_{u\in W_m}\mathcal{E}_u(\id_u\otimes  h_{(u,1)}\otimes \cdots h_{(u,k)}) =  h_m
\end{equation}
Let $a = \bigotimes_{u\in \Lambda_n}a_u$, for each subset $I\subseteq\Lambda_n$, we denote $a_{I} := \bigotimes_{u\in I}a_u \otimes \id_{W_n\setminus I}$.

For each $m\ge n+1$ one has
 \begin{eqnarray*}
      \varphi_{m}(a)&:=&\varphi_o\circ E_{ \Lambda_0}\circ E_{ \Lambda_1}\circ\ldots\circ {E}_{ \Lambda_{m}}(h_{m+1}^{1/2}a h_{m+1}^{1/2})\\
      &\overset{(\ref{E_Lbn})}{=}& \varphi_o\left(\mathcal{E}_{ W_0}\left(a_o\otimes \mathcal{E}_{W_1}\left(a_{W_n}\cdots \mathcal{E}_{ W_{n}}\left(a_{W_n}\otimes \mathcal{E}_{ W_{n+1}}\left(\id_{W_{n+1}}\cdots  \mathcal{E}_{ W_{m}}\left(\id_{W_m}\otimes h_{m+1}\right)\right) \right)\right)\right)\right)\\
     &\overset{(\ref{Ewm(h)})}{=}& \varphi_o\left(\mathcal{E}_{ W_0}\left(a_o\otimes \mathcal{E}_{W_1}\left(a_{W_n}\cdots \mathcal{E}_{ W_{n}}\left(a_{W_n}\otimes \mathcal{E}_{ W_{n+1}}\left(\id_{W_{n+1}}\cdots  \mathcal{E}_{ W_{m-1}}\left(\id_{W_m}\otimes h_{m}\right)\right) \right)\right)\right)\right))\\
      &&\vdots\\
      &=&\varphi_o\left(\mathcal{E}_{W_0}\left(a_{W_0} \dots\left( \mathcal{E}_{W_{n-1}}\left(a_{W_{n-1}}\left( \mathcal{E}_{W_n}\left(a_{W_n}\otimes h_{n+1} \right)\right)\right)\right)\right)\right)
  \end{eqnarray*}
  Then the limit (\ref{lim_Mc}) exists in the strongly finite sense and defines a positive functional $\varphi$.  Thanks to (\ref{TrwohoOQRW})  $\varphi$ is a state on $\mathcal{A}_V$.
Therefore, the triplet $(\omega_o, (\mathcal{E}_u)_{u\in V}, (h_u)_{u\in V})$ defines a quantum Markov chain $\varphi$ on the algebra $\mathcal{A}_V$ given by
\begin{equation}\label{expression_phi}
\varphi(a) =   \varphi_o\left(\mathcal{E}_{W_0}\left(a_{W_0} \dots\left( \mathcal{E}_{W_{n-1}}\left(a_{W_{n-1}}\left( \mathcal{E}_{W_n}\left(a_{W_n}\otimes h_{n+1} \right)\right)\right)\right)\right)\right)
\end{equation}
Now let us determine the expression for $\varphi$.
From the tree structure, we get
$$
\mathcal{E}_{W_n}(a_{W_n}\otimes h_{n+1})
= \bigotimes_{v\in W_n}\mathcal{E}_{v}(a_{v}\otimes h_{(v,1)}\otimes \cdots\otimes h_{(v,k)})
$$
and
$$
\mathcal{E}_{W_{n-1}}\left(a_{W_{n-1}}\otimes \mathcal{E}_{W_n}\left(a_{W_n}\otimes h_{n+1} \right)\right) = \bigotimes_{u\in W_{n-1}}\mathcal{E}_u\left(a_u\otimes\bigotimes_{v\in S(u)}\mathcal{E}_v\left(a_v\otimes h_{(v,1)}\otimes \cdots\otimes h_{(v,k)}\right)\right)
$$
For each $u\in W_{n-1}$, one finds
\begin{eqnarray}
&&\mathcal{E}_u\left(a_u\otimes\bigotimes_{v\in S(u)}\mathcal{E}_v\left(a_v\otimes h_{(v,1)}\otimes \cdots\otimes h_{(v,k)}\right)\right) \nonumber \\
&& \overset{(\ref{Eu})}{=} \mathcal{E}_u\left(a_u\otimes\bigotimes_{v\in W_n}\left(\sum_{(i_v,j_v,j_v')\in\Lambda^3} \prod_{\ell=1}^{k}\varphi_{j_v,j_v'}(h_{(v,\ell)}){M_{j_v}^{i_v}}^{*} a_v M_{j'_v}^{i_v}\right)\right) \nonumber \\
&& = \sum_{\mathbf{i,j,j'}\in\Lambda^{S(u)}}\left(\prod_{v\in W_n}\prod_{\ell=1}^{k}\varphi_{j_v,j_v'}(h_{(v,\ell)})\right)
\mathcal{E}_{u}\left(a_{u}\otimes\bigotimes_{v\in S(u)}{ M_{j_{v}}^{i_v}}^*a_vM_{j'_{v}}^{i_{v}} \right)\nonumber \\
&& = \sum_{\mathbf{i,j,j'}\in\Lambda^{S(u)}}\sum_{i_u, j_u, i'_u j'_u\in \Lambda}\left(\prod_{v\in W_n}\prod_{\ell=1}^{k}\varphi_{j_v,j_v'}(h_{(v,\ell)})\right)\Tr\left(A_{j_u}^{i_u}{ M_{j_{v}}^{i_v}}^*a_vM_{j'_{v}}^{i_{v}}A_{j'_u}^{i'_u\, *}\right) {M_{j_u}^{i_u}}^*a_uM_{j'_u}^{i'_u}  \label{Eu_sum}
\end{eqnarray}
where $\mathbf{i'}= (i_v)_{v\in S(u)},\, \mathbf{j} = (j_v)_{v\in S(u)},\, \mathbf{j'} = ( j'_v)_{v\in S(u)}$ are sequences of elements of $\Lambda$ induced by $S(u).$

According to (\ref{Aij}) and (\ref{Mij}), for each $v\in S(u)$, we obtain
\begin{eqnarray*}
  &&\sum_{i_v\in\Lambda}\Tr\left(A_{j_u}^{i_u}{ M_{j_{v}}^{i_v}}^*a_v M_{j'_{v}}^{i_{v}}A_{j'_u}^{i'_u\, *}\right) =
\Tr\left(M_{j'_{v}}^{i_{v}}A_{j'_u}^{i'_u\, *}A_{j_u}^{i_u}{ M_{j_{v}}^{i_v}}^*a_v\right)  \\
  && = \sum_{i_v\in\Lambda}\frac{1}{\sqrt{\Tr(\rho_{j_u})\Tr(\rho_{j'_u})}}
\Tr\left(B_{j'_{v}}^{i_{v}}\rho_{j'_u}^{1/2}\rho_{j_u}^{1/2}B_{j_{v}}^{i_{v}\,*}\otimes
 |i_v\rangle\langle j^{'}_v||j'_{u}\rangle\langle i'_u||i_{u}\rangle\langle j_u||j_{v}\rangle\langle i_v| a_v\right)  \\
   &&=\frac{1}{\sqrt{\Tr(\rho_{j_u})\Tr(\rho_{j'_u})}}
\sum_{i_v\in\Lambda}\Tr\left(B_{j'_{v}}^{i_{v}}\rho_{j'_u}^{1/2}\rho_{j_u}^{1/2}B_{j_{v}}^{i_{v}\,*}\otimes
 |i_v\rangle\langle i_v| a_v\right)\delta_{j_v,j_u}\delta_{j'_v, j'_u}\delta_{i'_u, i_u}\\
  &&=\frac{1}{\sqrt{\Tr(\rho_{j_v})\Tr(\rho_{j'_v})}}
\sum_{i_v\in\Lambda}\Tr\left(B_{j'_{v}}^{i_{v}}\rho_{j'_v}^{1/2}\rho_{j_v}^{1/2}B_{j_{v}}^{i_{v}\,*}\otimes
 |i_v\rangle\langle i_v| a_v\right)\delta_{j_v,j_u}\delta_{j'_v, j'_u}\delta_{i'_u, i_u}\\
  &&= \psi_{j_vj'_v}(a_v)\delta_{j_v,j_u}\delta_{j'_v, j'_u}\delta_{i'_u, i_u}.
\end{eqnarray*}
here, as before, $\psi_{j_u,j'_u}$ is given by (\ref{psi_jjprime}).

This implies that, for  any $u\in W_{n-1}$ among the configurations $\mathbf{j,j'}\in \Lambda^{\{u\}\cup S(u)}$ only ones satisfying $(j_u, j'_u) =( j_v, j'_v)$ for all $v\in S(u)$, appear on the sum of the right hand side of (\ref{Eu_sum}). It follows that (\ref{Eu_sum}) becomes
$$
 \sum_{i, j,j'\in \Lambda}\left(\prod_{v\in S(u)}\psi_{j,j'\in\Lambda}(a_v)\right)\left(\prod_{v\in S(u)}\prod_{\ell=1}^{k}\varphi_{j,j'}(h_{(v,\ell)})\right)  {M_{j}^{i}}^*a_uM_{j'}^{i}
$$
Iterating the above procedure, for each $m\le n$, one finds
$$
\mathcal{E}_{W_m}(a_{W_m}\otimes \mathcal{E}_{W_{m+1}}(a_{W_{m+1}}\otimes\cdots \mathcal{E}_{W_n}(a_{W_n}\otimes h_{n+1})))
$$
$$
= \sum_{j,j'\in \Lambda^{W_m}}\left(\prod_{v\in \Lambda_{[m,n]}}\psi_{j,j'\in\Lambda}(a_v)\right)\left(\prod_{w\in W_{n+1}}\prod_{\ell=1}^{k}\varphi_{j,j'}(h_{w})\right) \bigotimes_{u\in W_m}\sum_{i_u} {M_{j}^{i_u}}^*a_uM_{j'}^{i_u}
$$
Since for $m=0$, one has $W_0 = \{o\}$ then
$$
\varphi(a) = \sum_{j,j'\in \Lambda }\left(\prod_{v\in \Lambda_{n}}\psi_{j,j'\in\Lambda}(a_v)\right)\left(\prod_{w\in W_{n+1}}\prod_{\ell=1}^{k}\varphi_{j,j'}(h_{w})\right) \varphi_o\bigg(\sum_{i\in\Lambda} {M_{j}^{i}}^*a_uM_{j'}^{i}\bigg)
$$
hence
 $$
 \varphi_o\bigg(\sum_{i \in\Lambda} {M_{j}^{i }}^*a_uM_{j'}^{i }\bigg) = \Tr\bigg(\omega_o\sum_{i\in\Lambda}{M_{j}^{i }}^*a_o M_{j'}^{i }\bigg)
 $$
 $$=\sum_{i\in\Lambda} \Tr(\omega_o{M_{j}^{i}}^*a_oM_{j'}^{i }) = \sum_{i\in \Lambda } \Tr(M_{j'}^{i }\omega_o{M_{j}^{i}}^*a_o) = \Tr(\mathcal{M}_{jj'}(\omega_o)a_o)
 $$
  where $\mathcal{M}_{jj'}$ is given by (\ref{Mjj'}). This completes the proof.
  \end{proof}

\begin{remark}
The maps $\varphi_{jj'}$ and $\psi_{jj'}$ are linear functionals. If the particular case  $j=j'$, the linear functionals $\varphi_{jj}$ and $\psi_{jj}$ define two states, and we have
\begin{equation}\label{phi_j}
  \varphi_{jj}(b) =\frac{1}{\Tr(\rho_j)}\Tr\left(\rho_j \otimes |j\rangle\langle j|b\right).
\end{equation}
and
\begin{equation}\label{psi_j}
  \psi_{jj}(b) = \frac{1}{\Tr(\rho_j)}\sum_{i\in\Lambda}\Tr\left(B_{j}^{i}\rho_j{B_{j}^{i}}^{*}\otimes |i\rangle\langle i|b\right).
\end{equation}
\end{remark}

\subsection{Quantum Markov chain associated with the disordered phase}

In this section, we are going to discuss about QMC associated with the disordered phase of the system. Here, by the disordered phase, it is meant a QMC corresponding to the trivial solution of (\ref{h=E(h)OQRW}).  Indeed, we have the following result.

\begin{theorem}
Assume that Hilbert spaces $\mathcal{H}$ and $\mathcal{K}$ are finite dimensional, then $h_0 = \id_{\mathcal{B}(\mathcal{H})}\otimes\id_{\mathcal{B}(\mathcal{K})}$ defines a  homogeneous boundary condition   $\mathbf{h}_0:= (h_u = \alpha_u(h_0))_{u\in V}$ satisfying (\ref{h=E(h)OQRW}). Moreover for any normalized density matrix $\omega_o\in\mathcal{A}_{o;+}$  the quantum Markov chain $\varphi$ associated with the triplet $(\omega_o, \mathcal{E}, h_0)$ is given by
\begin{equation}\label{varphi_id}
    \varphi(a) = \sum_{j} \Tr\left(\mathcal{M}_{jj}(\omega_o) a_o \right)\prod_{u\in \Lambda_{[1,n]}}\psi_{j,j}(a_{u})
\end{equation}
for every $a=\bigotimes_{u\in\Lambda_n}a_u\in\mathcal{A}_{\Lambda_n}$.

\end{theorem}

\begin{proof}
For each $u\in V$ and $j,j'\in \Lambda$, one has
 $$
 \varphi_{jj'}(\id):= \frac{1}{\Tr(\rho_j)^{1/2}\Tr(\rho_{j'})^{1/2}}\Tr\left(\rho_j^{1/2} \rho_{j'}^{1/2} \otimes |j'\rangle\langle j|\right) = \delta_{jj'}
 $$
 From (\ref{Eu}), it follows that
\begin{eqnarray*}
\mathcal{E}_u( \id^{(u)}\otimes\id^{(u,1)}\otimes\cdots\otimes \id^{(u,k)} ) &=&   \sum_{(i,j ,j')\in \Lambda^3}   M_j^{i*}M_{j'}^{i}  \delta_{jj'} \\
&=& \sum_{(i,j )\in \Lambda^2}   M_j^{i*}M_{j }^{i}\\
&\overset{(\ref{sumMij=1})}{=}& \id^{(u)}.
\end{eqnarray*}
 This proves that $\mathbf{h}_0$ defines a boundary condition. Let $\varphi$ be a QMC associated with this boundary condition. From  (\ref{phi_id})  one gets (\ref{varphi_id}).
\end{proof}

\begin{remark}
We notice that QMC given in (\ref{phi_id})  generalizes the Markov chains associated with open quantum random walks studied in \cite{DM19}  to trees.
 In the one dimensional setting, they propose a more general class of QMC associated with OQRW than the ones considered in \cite{DM19} due to the existence of boundary conditions.
\end{remark}

   \section{Phase transition for QMCs on  trees associated with two-state OQRWs}\label{Sect_PT}

   In this section,  we focus on several applications of quantum Markov chains associated open quantum random walks on trees to a phenomena of phase transition, the following definition of phase transition within QMC scheme was first introduced in \cite{MBS161}.

 \begin{definition}\label{def_PT}
   We say that there exists a phase transition for the constructed QMC associated with OQRW if the following conditions are
satisfied:
\begin{enumerate}
\item[(a)] {\sc existence}: The equations \eqref{TrwohoOQRW}, \eqref{h=E(h)OQRW}
have at least two $(u_0,\{h^x\}_{x\in L})$ and $(v_0,\{s^x\}_{x\in
L})$ solutions;
\item[(c)] {\sc not overlapping supports}: there is a projector
$P\in B_L$ such that $\ffi_{u_0,\bh}(P)<\ve$ and
$\ffi_{v_0,\bs}(P)>1-\ve$, for some $\ve>0$.
\item[(b)] {\sc not quasi-equivalence}: the corresponding quantum
Markov chains $\ffi_{u_0,\bh}$ and $\ffi_{v_0,\bs}$ are not quasi
equivalent
\end{enumerate}
Otherwise, we say there is no phase transition.
\end{definition}

In this section, for the sake of simplicity, we reduce the study to the case   $\mathcal{H}= \mathcal{K} = \mathbb{C}^2$. For each $u\in V$, we take  $\mathcal{A}_u = \mathcal{B}(\mathcal{H}\otimes\mathcal{H}) \equiv M_4(\mathbb{C})$ and let $\Lambda = \{1,2\}$. The interactions are given by
  \begin{equation}\label{Bjimodel}
     B_1^1 = \left(
               \begin{array}{cc}
                 a & 0 \\
                 0 & b \\
               \end{array}
             \right), \quad  B_2^1 = \left(
               \begin{array}{cc}
                 0 & 1 \\
                 0 & 0 \\
               \end{array}
             \right), \quad  B_1^2 = \left(
               \begin{array}{cc}
                 c & 0 \\
                 0 & d \\
               \end{array}
             \right),\quad  B_2^2 = \left(
               \begin{array}{cc}
                 1 & 0 \\
                 0 & 0 \\
               \end{array}
             \right)
   \end{equation}
   where  $|a|^2 + |c|^2 = |b|^2 + |d|^2  =1, ac\ne 0$.
   Put \begin{equation}\label{pq}
p=\left(
      \begin{array}{cc}
        1 & 0 \\
       0 & 0 \\
      \end{array}
    \right), \qquad q=\left(
      \begin{array}{cc}
        0 & 0 \\
       0 & 1 \\
      \end{array}
    \right).
\end{equation}
and
$$|1\rangle = \left[\begin{array}{cc} 1 \\ 0  \end{array}\right] , |2\rangle  = \left[\begin{array}{ll}0 \\ 1  \end{array}\right] $$
 Notice that  $(|1\rangle,|2\rangle)$ is an ortho-normal basis of $\mathcal{K}\equiv \mathbb{C}^2$. In  the sequel elements of $\mathcal{B}(\mathcal{H})$ will be denoted by means of $2\times2$ complex matrices, while elments of $\mathcal{B}(\mathcal{K})$ will be written using Dirac notation $|i><j|.$

\subsection{Existence of boundary conditions and their associated quantum Markov chains}
In this section, we are goong to determine all the translation invariant boundary conditions  associated with the considered two-state  OQRW (\ref{Bjimodel}). This means that we describe positive solutions $h\in\mathcal{A}_{u,;+}$ of the compatibility equation (\ref{h=E(h)OQRW}).

\begin{lemma}\label{invariant_boundary_conditions1}
The translation  invariant boundary conditions associated with the two-states OQRW  (\ref{Bjimodel}) are given by:
\begin{equation}\label{hblocks}
    h^{(u)}=\sum_{j,j'\in\Lambda}h_{j,j'}\otimes |j\rangle\langle j'|\in \mathcal{B}(\mathcal{H})\otimes \mathcal{B}(\mathcal{K}),
\end{equation}
where
\begin{equation*}
     h_{j,j'}=\left( \frac{\Tr( \rho_{j'}^{1/2}\rho_j^{1/2} h_{j,j'} )}{\sqrt{\Tr(\rho_j)\Tr(\rho_j')}}\right)^k
\sum_{i\in\Lambda}{B_{j}^{i*}}{B_{j'}^{i}}.
 \end{equation*}
 \end{lemma}
 \begin{proof}
From (\ref{h=E(h)OQRW}), one has
$$
h^{(u)} = \sum_{i,j,i',j'=1,2}{M_{j}^{i*}}^{(u)}{M_{j'}^{i'}}^{(u)} \prod_{\ell=1}^{k}\Tr(  A_{j}^{i}h^{(u,\ell)} {A_{j'}^{i'}}^{*}).$$
Since the boundary condition is translation invariant (i.e. $h^{(u)} = h $ for all $u\in V$), one gets
\begin{equation}\label{Eq1}
h = \sum_{i,j,i',j'=1,2}{M_{j}^{i*}}{M_{j'}^{i'}} \Tr(  A_{j}^{i}h {A_{j'}^{i'}}^{*})^k.
\end{equation}
Now, using
 \[
 {M_{j}^{i*}}{M_{j'}^{i'}}={B_{j}^{i*}}{B_{j'}^{i'}}\otimes |j\rangle\langle j'|\delta_{i,i'}
 \]
 and
 $$
 \Tr( A_{j}^{i}h {A_{j'}^{i'}}^{*})=\Tr( {A_{j'}^{i'}}^{*}A_{j}^{i}h )   =\frac{1}{\sqrt{\Tr(\rho_j)\Tr(\rho_j')}} \Tr( \rho_{j'}^{1/2}\rho_j^{1/2}\otimes|j'\rangle\langle j| h )\delta_{i,i'},
 $$
 The \eqref{Eq1} becomes
 \begin{equation*}
h = \sum_{i,j,j'\in\Lambda}\left( \frac{\Tr( \rho_{j'}^{1/2}\rho_j^{1/2}\otimes|j'\rangle\langle j| h )}{\sqrt{\Tr(\rho_j)\Tr(\rho_j')}}\right)^k
{B_{j}^{i*}}{B_{j'}^{i}}\otimes |j\rangle\langle j'|.
\end{equation*}
 Finally, by identification with (\ref{hblocks}), we are led to
 \begin{equation*}
     h_{j,j'}=\sum_{i\in\Lambda}\left( \frac{\Tr( \rho_{j'}^{1/2}\rho_j^{1/2}\otimes|j'\rangle\langle j| h )}{\sqrt{\Tr(\rho_j)\Tr(\rho_j')}}\right)^k
{B_{j}^{i*}}{B_{j'}^{i}}=\left( \frac{\Tr( \rho_{j'}^{1/2}\rho_j^{1/2} h_{j,j'} )}{\sqrt{\Tr(\rho_j)\Tr(\rho_j')}}\right)^k
\sum_{i\in\Lambda}{B_{j}^{i*}}{B_{j'}^{i}}.
 \end{equation*}
 \end{proof}

 \begin{theorem}\label{invariant_boundary_conditions2}
 Let $\{E_{jj'},\ 1\ge j,j'\le 4\}$ denotes the canonical basis of $\mathcal{B}(\mathcal{H})\otimes \mathcal{B}(\mathcal{K})\equiv M_4(\mathbb{C})$.
 Write the  translation  invariant boundary conditions associated with the two-states OQRW  (\ref{Bjimodel}) in the form
 \begin{equation*}
      h =\sum_{j,j'=1}^4e_{jj'}E_{jj'}.
 \end{equation*}
 Then
 \begin{equation}\label{sys}
\left\{
     \begin{array}{ccc}
       e_{11}=&e_{33} =& \frac{\Tr(\rho_1\otimes|1\rangle\langle 1| h_{1,1} )^k}{\Tr(\rho_1)^k}\\
       e_{22} =&e_{44} =&\frac{\Tr(\rho_{2}\otimes|2\rangle\langle 2| h_{2,2} )^k}{\Tr(\rho_2)^k}\\
        e_{12} = &\frac{\overline{c}}{\overline{a}}e_{14}=&\overline{c}\frac{\Tr( \rho_{2}^{1/2}\rho_1^{1/2}\otimes|2\rangle\langle 1| h_{1,2} )^k}{\left(\Tr(\rho_1)\Tr(\rho_2)\right)^{k/2}}\\
        e_{21}=&\frac{c}{a}e_{41}=& c\frac{\Tr( \rho_{1}^{1/2}\rho_2^{1/2}\otimes|1\rangle\langle 2| h_{2,1} )^k}{\left(\Tr(\rho_1)\Tr(\rho_2)\right)^{k/2}}\\
        e_{j,j'}=&0,& otherwise
     \end{array}
   \right.
\end{equation}
In particular, for the initial states $\rho_1=\rho_2=|1\rangle\langle 1|$ and $k=2$, there are exactly $4$ non--trivial solutions given by
\begin{equation}\label{hi_model}
h_0 = \id_{\mathbb{M}_4},\quad h_1= \id_{M_2}\otimes |1\rangle\langle 1|,\quad h_2= \id_{M_2}\otimes |2\rangle\langle 2|, \quad
 h_3= h_0+h_c,
\end{equation}
where
\begin{align*}
    h_c=&\frac{1}{\overline{c}}B_2^2\otimes B_2^1
+\frac{1}{c}B_2^2\otimes {B_2^1}^*\quad \hbox{defined only if}\quad |c|=1,
\end{align*}
Moreover,  if for each $i\in \{0,1,2,3\}$ an initial density matrix $\omega_i$ associated with the boundary condition $h_i$ satisfying (\ref{TrwohoOQRW}) then there exists four quantum Markov chains  $\varphi_{\omega_j,h_j}, 0\le i\le3$  and
\begin{equation}\label{phi_hj}
\varphi_{\omega_j,h_j}(a)=  \Tr\left(\mathcal{M}_{jj}(\omega_o) a_o \right)\prod_{u\in \Lambda_{[1,n]}}\psi_{j,j}(a_{u}),  \quad \forall j\in \{1,2\}
\end{equation}
 \end{theorem}

 \begin{proof}
From Lemma \ref{invariant_boundary_conditions1} and the fact that $|a|^2+ |c|^2=|b|^2+ |d|^2=1$, a straightforward computation leads to \eqref{sys}.
Next, for the initial states $\rho_1=\rho_2=|1\rangle\langle 1|$ and $k=2$, the system \eqref{sys} reduces to,
\begin{equation*}
\left\{
     \begin{array}{ccc}
       e_{11}=&e_{33} =& e_{11}^2\\
       e_{22} =&e_{44} =&e_{22}^2\\
        e_{12} = &\frac{\overline{c}}{\overline{a}}e_{14}=&\overline{c}\,e_{12}^2\\
         e_{21}=&\frac{c}{a}e_{41}=& c\,e_{21}^2\\
         e_{jj'}=&0,& otherwise
     \end{array}
   \right.
\end{equation*}
for $a\ne0$ and
\begin{equation*}
\left\{
     \begin{array}{ccc}
       e_{11}=&e_{33} =& e_{11}^2\\
       e_{22} =&e_{44} =&e_{22}^2\\
        e_{12} =&\overline{c}\,e_{12}^2\\
         e_{21}=& c\,e_{21}^2\\
         e_{jj'}=&0,& otherwise
     \end{array}
   \right.
\end{equation*}
when $a=0$ (and so $|c|=1$).
Therefore, the non--negative Hermitian solutions satisfy
\begin{equation*}
\left\{
     \begin{array}{ccc}
       e_{11}=&e_{33} & \in\{0,1\}\\
       e_{22} =&e_{44} &\in\{0,1\}\\
        e_{12} =&\overline{e_{21}}&\in\{0,\frac{1}{\overline{c}}\}\\
         e_{jj'}=&0,& otherwise
     \end{array}
   \right.
\end{equation*}
This leads to the solutions (\ref{hi_model}). \\
Now since  for $i$ each the solutions $h_i$ is  positive there exists an initial $\omega_i\in \mathcal{A}_o$ such that $\Tr(\omega_ih_i) = 1$. Let $\mathcal{E}_u$ given by (\ref{Eu_def}) then from Theorem \ref{thm_E_expression} the triplet $(\omega_i, (\mathcal{E}_u)_u, h_i)$ defines a quantum Markov chain $\varphi_{\omega_i, h_i}$ on $\mathcal{A}_{V}$.
One can easily check that
\begin{align*}
\varphi_{j,j'}(h_1)=\delta_{j,1}\delta_{1,j'},\quad \varphi_{j,j'}(h_2)=\delta_{j,2}\delta_{2,j'},
\end{align*}
Therefore, according to (\ref{phi_id}) one gets (\ref{phi_hj}).
\end{proof}

 \subsection{Not quasi-equivalence property}
 Recall that two states $\varphi$ and $\psi$ on a $C^*$-algebra $\mathcal{A}$ are said be \textit{quasi-equivalent} if the GNS representations $\pi_\varphi$ and $\pi_\psi$ are quasi-equivalent. The reader is referred to  \cite{BR} for the notion of quasi-equivalence of representations. The following result proposes a criteria  for  the non-quasi equivalence, we are going to use the following result
(see \cite[Corollary 2.6.11]{BR}).
\begin{lemma}\label{br-q}
Let $\varphi_1,$ $\varphi_2$ be two factor states on a quasi-local
algebra $\ga=\cup_{\Lambda}\ga_\Lambda$. The states $\varphi_1,$
$\varphi_2$ are  quasi-equivalent if and only if for any given
$\varepsilon>0$ there exists a finite volume $\Lambda\subset V$
such that $|\varphi_1(a)-\varphi_2(a)|<\varepsilon \|a\|$ for
all $a\in B_{\Lambda^{'}}$ with
$\Lambda^{'}\cap\Lambda=\varnothing$.
\end{lemma}

\begin{theorem}
Assume that $|c|=1$. If $\mathcal{M}_{11}(\omega_1)=\mathcal{M}_{22}(\omega_2)$ then the quantum Markov chains $\varphi_{h_1, \omega_1}$ and $ \varphi_{h_2, \omega_2}$ are  quasi-equivalent.
\end{theorem}
\begin{proof} First from (\ref{phi_hj}) the two states $\varphi_{\omega_1,h_1}$ and $\varphi_{\omega_2,h_2}$ are product states then  they are factor states.
The assumption $|c|=1$ implies $|a|=0$, then  according to (\ref{phi_j}), (\ref{psi_j}) and (\ref{Bjimodel}) the two  states $\psi_{11}$ and $\psi_{22}$ coincide on $\mathcal{B}(\mathcal{H})\otimes\mathcal{B}(\mathcal{K})$.
It follows that,
\begin{equation*}
  |\varphi_{\omega_1, h_1}(a)- \varphi_{\omega_2, h_2}(a)| =| \Tr\left(\left(\mathcal{M}_{11}(\omega_1) -\mathcal{M}_{22}(\omega_2)\right) a_o \right)|\prod_{u\in \Lambda_{[1,n]}}|\psi_{1,1}(a_{u})|
\end{equation*}
Hence, clearly if we have $\mathcal{M}_{11}(\omega_1)=\mathcal{M}_{22}(\omega_2)$, then the quantum Markov chains $\varphi_{h_1, \omega_1}$ and $ \varphi_{h_2, \omega_2}$ are  quasi-equivalent.\\
\end{proof}

 \begin{theorem}
Assume that  $|c|<1$. The quantum Markov chains $\varphi_{h_1, \omega_1}$ and $ \varphi_{h_2, \omega_2}$ are not quasi-equivalent.
 \end{theorem}
 \begin{proof}
 Let us define an element of $\mathcal{A}_{\Lambda_n}$ as follows
 \begin{equation*}
  E_{\Lambda_n} = \sigma^{x_{W_n}(1)}\otimes \id_{\Lambda_{n}\setminus \{x_{W_n}(1)\}}
 \end{equation*}
where
\begin{equation*}
    \sigma^{x_{W_n}(1)} =  \id_{M_2}\otimes p
\end{equation*}
 here $x_{W_n}(1)$ is defined by (\ref{xw}). Then, we have
 \begin{equation*}
     \psi_{11}(\sigma^{x_{W_n}(1)})=\Tr(B_1^1p{B_1^1}^*)=|a|^2\quad{\rm and} \ \psi_{22}(\sigma^{x_{W_n}(1)})=\Tr(B_2^1p{B_2^1}^*)=0.
 \end{equation*}
 On the other hand,
 \begin{equation*}
   \psi_{j,j}(\id_{M_2}\otimes \id_{M_2})= \sum_{i=1}^{2}\Tr\left(p{B_j^i}^*
    B_{j'}^i\right)=\Tr\left(p\sum_{i=1}^{2}{B_j^i}^*
    B_{j}^i\right)=\Tr(p)=1.
\end{equation*}
Hence,
\begin{equation*}
\varphi_{\omega_1,h_1}(E_{\Lambda_n})=  \Tr\left(\mathcal{M}_{11}(\omega_1)  \right)|a|^2=|a|^2\quad{\rm and} \ \varphi_{\omega_2,h_2}(E_{\Lambda_n})=0.
\end{equation*}
 One gets
\begin{eqnarray*}
&& \varphi_{\omega_2,h_2}(E_{\Lambda_n}) = 0\\
&& \varphi_{\omega_1,h_1}(E_{\Lambda_n})=|a|^2.
\end{eqnarray*}
Now, since $|c|<1$, and $\|E_{\Lambda_n}\| = 1$ then $\varepsilon_0: =|a|^2>0 $
$$
|\varphi_{\omega_2,h_2}(E_{\Lambda_n}) - \varphi_{\omega_1,h_1}(E_{\Lambda_n})|= |a|^2 \ge \varepsilon_0\|E_{\Lambda_n}\|
$$
for every $n\ge0.$ Therefore from Lemma \ref{br-q} the two quantum Markov chains $ \varphi_{\omega_1,h_1}$ and $\varphi_{\omega_2,h_2}$ are not quasi-equivalent.
 \end{proof}
\subsection{Not overlapping support}

Recall that two states $\varphi$ and $\psi$ on $\mathcal{A}_{V}$ have not overlapping supports if there is a projector
$P\in B_V$ such that $\varphi(P)<\varepsilon$ and
$\psi(P)>1-\varepsilon$, for some $\varepsilon>0$.

\begin{lemma}\label{rank-1_proj}
Any rank-1 projection in $\mathbb{M}_2(\mathbb{C})$ has the form
\begin{equation}\label{ez}
 p(\varepsilon,z) = \left(\begin{array}{cc}  \varepsilon & z\sqrt{\varepsilon(1-\varepsilon)}\\
 \\
  \overline{z} \sqrt{\varepsilon(1-\varepsilon)}   & 1-\varepsilon \\
 \end{array}\right)
 \end{equation}
 where $\varepsilon\in[0,1], z\in \mathbb{C}$ with $|z| =1$.\\
 \end{lemma}

 \begin{remark}\label{M_jjexpression}
  Notice that,
 \begin{eqnarray*}
 \Tr(\mathcal{M}_{jj}(w))& =& \sum_{i\in\Lambda}\Tr(M_{j}^{i}w{M_{j}^{i}}^*) = \sum_{i\in\Lambda}\Tr({M_{j}^{i}}^*M_{j}^{i}w)\\[2mm]
  &=& \Tr(\sum_{i\in\Lambda}{B_{j}^{i}}^*B_{j}^{i}\otimes|j\rangle\langle j|w)\\
  & =&\Tr(\id_{\mathcal{H}}\otimes|j\rangle\langle j|w) .
\end{eqnarray*}
 Writing $w$ as a block matrix:
 \begin{equation*}
     w=\sum_{j,j'\in\lambda}\omega_{j,j'}\otimes|j\rangle\langle j'|,
 \end{equation*}
 one gets
 \begin{equation*}
     \Tr(\mathcal{M}_{jj}(w))=\Tr(\omega_{j,j}).
 \end{equation*}
 More generally, one has
 \begin{equation*}
     {M}_{jj'}(w)=\sum_{i\in\Lambda}M_{j'}^{i} w {M_{j}^{i}}^* = \sum_{i\in\Lambda}B_{j'}^{i} \omega_{j',j} {B_{j}^{i}}^*\otimes|i\rangle\langle i|.
 \end{equation*}
 \end{remark}
  For each $n\in\mathbb{N}$, we denote
$$
p_n(\varepsilon,z) = \bigotimes_{u\in\Lambda_n}(p(\varepsilon,z)\otimes I)^{(u)} \in\mathcal{A}_{\Lambda_n}
$$

 \begin{theorem}\label{notovelapping}
 The quantum Markov chains $ \varphi_{\omega_1,h_1}$ and $\varphi_{\omega_2,h_2}$ have not overlapping supports if and only if $\Tr(\omega_1 p\otimes \id_{M_2})\ne1$.
 \end{theorem}
\begin{proof}
One has
$$\varphi_{\omega_1,h_1}(p_n(\varepsilon,z))= \varepsilon^{2^{n+1}-2} \Tr\left(\mathcal{M}_{11}(\omega_1) p(\varepsilon,z)\otimes\id_{M_2} \right),
$$
and
$$
\varphi_{\omega_2,h_2}(p_n(\varepsilon,z))= \varepsilon^{2^{n+1}-2} \Tr\left(\mathcal{M}_{22}(\omega_2) p(\varepsilon,z)\otimes\id_{M_2} \right),
$$
Let $P_n=p_n(\varepsilon,z)$ be a rank-1 projector.
\begin{itemize}
    \item If $\varepsilon<1$, then one has
\begin{equation*}
\lim_{n\rightarrow\infty}\varphi_{\omega_1,h_1}(P_n)=\lim_{n\rightarrow\infty}\varphi_{\omega_2,h_2}(P_n)=0.
\end{equation*}
    \item If $\varepsilon=1$, then using Remark \ref{M_jjexpression} one has
\begin{equation*}
\varphi_{\omega_2,h_2}(P_n)=\Tr(\omega_2h_2)=1.
\end{equation*}
On the other hand one has
\begin{equation*}
\varphi_{\omega_1,h_1}(P_n)=\Tr(\omega_1 p\otimes \id_{M_2})\ne1=\varphi_{\omega_2,h_2}(P_n).
\end{equation*}
\end{itemize}
\end{proof}
The above results leads to the following concluding result.

\begin{theorem}
if $|c|<1$, then there exists a phase transitions for the quantum Markov chains associated with the two-state OQRW (\ref{Bjimodel}).
\end{theorem}
\begin{proof}
Take
\begin{equation*}
    \omega_1=q\otimes|1\rangle\langle1|\quad {\rm and}\ \omega_2=p\otimes|2\rangle\langle2|.
\end{equation*}
Using (\ref{Mjj'}), (\ref{Bjimodel}) and (\ref{Mij}) we find
\begin{equation*}
    \mathcal{M}_{11}(\omega_1)=q\otimes\Big(|b|^2|1\rangle\langle1| + |d|^2|2\rangle\langle2|\Big)
\end{equation*}
and
\begin{equation*}
   \mathcal{M}_{22}(\omega_2)=p\otimes|2\rangle\langle2|.
\end{equation*}
On the one hand, by Lemma \ref{br-q}, the two quantum Markov chains $ \varphi_{\omega_1,h_1}$ and $\varphi_{\omega_2,h_2}$ are not quasi-equivalent since
\[
|\varphi_{\omega_2,h_2}(E_{\Lambda_n}) - \varphi_{\omega_1,h_1}(E_{\Lambda_n})|= |a|^2
\]
for every $n\ge0.$
On the other hand, clearly (by Theorem \ref{notovelapping}) the two quantum Markov chains $ \varphi_{\omega_1,h_1}$ and $\varphi_{\omega_2,h_2}$ are not overlapping supports since
\[\Tr(\omega_1( p\otimes \id_{M_2}))=0\ne1.\]
\end{proof}

\section{Mean entropy for QNCs  on trees associated with OQRWs}

This section is devoted to the computation of mean entropies for the two quantum Markov chains $\varphi_{\omega_1, h_1}$ and $\varphi_{\omega_2, h_2}$ given by $(\ref{phi_hj})$. In the notations of section  \ref{Sect_PT}, the algebra $\mathcal{A}_x \equiv \mathbf{M}_4(\mathbb{C})$. Let $\varphi$ be a state on $\mathcal{A}_V$.
For each bounded region $I$ of the vertex set $V$, the density matrix of the restriction $\varphi\lceil_{\mathcal{A}_I}$ will be denoted by $D^{\varphi}_I$. The von Neumann entropy of   $\varphi\lceil_{\mathcal{A}_I}$  is defined to be
\begin{equation}
    S(\varphi)  =  - \Tr(D_\varphi\log D_\varphi)
\end{equation}
In \cite{MS21}  the mean entropy for quantum Markov states on trees was defined  as follows
     \begin{equation}\label{mean_entropy_tree}
     s(\varphi) := \lim_{n\to\infty}\frac{1}{|\Lambda_n|}S(\varphi\lceil_{\mathcal{A}_{\Lambda_n}})
 \end{equation}
 We use the same formulae for the quantum Markov chains $\varphi_{\omega_1, h_1}$ and $\varphi_{\omega_2, h_2}$.
 \begin{theorem} For each $j\in\{1,2\},$
 let $\varphi_{\omega_j, h_j}$ be given by (\ref{phi_hj}). The mean entropy of the quantum Markov chin $\varphi_{\omega_j, h_j}$ coincides with the von Neumann entropy of the state $\psi_{jj}$:
 \begin{equation}
     s(\varphi_{\omega_j,h_j}) = S(\psi_{jj})
 \end{equation}
 where $\psi_{jj}$ is as in (\ref{psi_jjprime}). Therefore, one has
\begin{equation*}
     s(\varphi_{\omega_1,h_1}) = -2\bigg(|a|^2\log|a| + |c|^2\log|c|\bigg)
 \end{equation*}
 and
 \begin{equation*}
     s(\varphi_{\omega_2,h_2}) = 0
 \end{equation*}
 \end{theorem}
 \begin{proof}

 From (\ref{phi_hj}) the density matrices of $\varphi_{\omega_j, h_j}$ is given as follows
 \begin{equation}
     D_{\varphi_{\omega_j, h_j}\lceil_{\mathcal{A}_{\Lambda_n}}} = \mathcal{M}_{jj}(\omega_1)^{(o)} \otimes\bigotimes_{x\in\Lambda_{[1,n]}} D_{\psi_{jj}}^{(x)}
 \end{equation}
 One finds
 \begin{equation}\label{logDphi}
 \log  D_{\varphi_{\omega_j, h_j}\lceil_{\mathcal{A}_{\Lambda_n}}} =  \log \mathcal{M}_{jj}(\omega_1)^{(o)} + \sum_{x\in \Lambda_{[1,n]}}\log D_{\psi_{jj}}^{(x)}
 \end{equation}
 One has
 \begin{eqnarray*}
 S\left(\varphi_{\omega_j, h_j}\lceil_{\mathcal{A}_{\Lambda_n}}\right) &= & -\Tr\left(D_{\varphi_{\omega_j, h_j}\lceil_{\mathcal{A}_{\Lambda_n}}}\log D_{\varphi_{\omega_j, h_j}\lceil_{\mathcal{A}_{\Lambda_n}}} \right)\\
 &=& - \varphi_{\omega_j, h_j}\left( \log D_{\varphi_{\omega_j, h_j}\lceil_{\mathcal{A}_{\Lambda_n}}}\right)\\
 &\overset{(\ref{logDphi})}{=} & -  \varphi_{\omega_j, h_j}\left(\log \mathcal{M}_{jj}(\omega_1)^{(o)}\right) - \sum_{x\in \Lambda_{[1,n]}}\varphi_{\omega_j, h_j}\left(\log D_{\psi_{jj}}^{(x)} \right)\\
 &\overset{(\ref{psi_jjprime})}{=} & -\Tr\left(\mathcal{M}_{jj}(\omega_1)^{(o)}\log \mathcal{M}_{jj}(\omega_1)^{(o)}\right) -  \sum_{x\in \Lambda_{[1,n]}}\psi_{jj}\left(\log D_{\psi_{jj}}^{(x)} \right)\\
 &= & -\Tr\left(\mathcal{M}_{jj}(\omega_1)^{(o)}\log \mathcal{M}_{jj}(\omega_1)^{(o)}\right) -  |\Lambda_{[1,n]}|\psi_{jj}\left(\log D_{\psi_{jj}}\right)\\
 &=&  S(\varphi_o) + (2^{n+1}-2) S(\psi_{jj})
 \end{eqnarray*}
It follows that the mean entropy of the quantum Markov chains $\varphi_{\omega_j, h_j}$ is given by
\begin{eqnarray*}
    s(\varphi_{\omega_j, h_j}) &=&\lim_{n\to \infty} \frac{ S\left(\varphi_{\omega_j, h_j}\lceil_{\mathcal{A}_{\Lambda_n}}\right)}{|\Lambda_n|}\\
    &=&  \lim_{n\to \infty} \left( -\frac{S(\varphi_o)}{2^{n+1}-1} - \frac{2^{n+1}-2}{2^{n+1}-1}S(\psi_{jj}) \right) \\
    &=& S(\psi_{jj})
\end{eqnarray*}
Now,  from (\ref{psi_jjprime}) one can see that the density matrix of $\psi_{jj}$ is defined by
 $$
 D_{\psi_{jj}} = \frac{1}{\Tr(\rho_j)} \sum_{i\in \Lambda}B_{j}^{i}\rho_jB_{j}^{i*}\otimes |i\rangle\langle i|=\sum_{i\in \Lambda}B_{j}^{i}pB_{j}^{i*}\otimes |i\rangle\langle i|
 $$
since $\rho_j=|1\rangle\langle1|=p$. It follows that
$$
 D_{\psi_{11}} =p \otimes \bigg(|a|^2 |1\rangle\langle 1|+   |c|^2   |2\rangle\langle 2|\bigg)
 $$
 and
 $$
 D_{\psi_{22}} = q \otimes |2\rangle\langle 2|
 $$
 Hence we get
 \begin{equation*}
     s(\varphi_{\omega_1,h_1})=  - \Tr(D_{\psi_{11}}\log D_{\psi_{11}}) = -2\bigg(|a|^2\log|a| +|c|^2\log|c|\bigg)
 \end{equation*}
 and
 \begin{equation*}
     s(\varphi_{\omega_2,h_2}) =  - \Tr(D_{\psi_{22}}\log D_{\psi_{22}}) = 0
 \end{equation*}
 That completes the proof.
\end{proof}

\section*{Acknowledgments}
The authors gratefully acknowledge Qassim University, represented by the Deanship of Scientific Research, on the
financial support for this research under the number (10173-cba-2020-1-3-I)
during the academic year 1442 AH / 2020 AD.

\end{document}